\documentclass[a4paper,12pt]{article}

\usepackage[utf8]{inputenc}
\usepackage[backend=biber,style=alphabetic,maxbibnames=100,maxalphanames=4]{biblatex}
\usepackage[in]{fullpage}
\usepackage[dvipsnames,svgnames,x11names]{xcolor}
\usepackage[colorlinks]{hyperref}
\usepackage{enumerate}
\usepackage{nicefrac}
\usepackage{graphicx}
\usepackage{algorithm}
\usepackage{amsmath}
\usepackage{mathtools}
\usepackage{authblk}
\usepackage{amssymb}
\usepackage{amsthm}
\usepackage{esvect}
\usepackage[normalem]{ulem}
\usepackage{braket}
\usepackage{pgfplots}
\usepackage{multicol}
\usepackage[final]{microtype}
\usepackage[font=small]{caption}
\usetikzlibrary{calc}

\makeatletter
\newtheorem*{rep@theorem}{\rep@title}
\newcommand{\newreptheorem}[2]{%
\newenvironment{rep#1}[1]{%
 \def\rep@title{#2 \ref{##1}}%
 \begin{rep@theorem}}%
 {\end{rep@theorem}}}
\makeatother
\newreptheorem{theorem}{Theorem}

\pgfplotsset{compat=1.6}

\bibliography{bibliography}

\hypersetup{
    breaklinks = true,
	colorlinks = true,
	citecolor = teal,
	urlcolor = purple,
}

\newtheorem{theorem}{Theorem}
\newtheorem{lemma}[theorem]{Lemma}

\newtheorem{proposition}[theorem]{Proposition}

\DeclareMathOperator{\codim}{codim}
\DeclareMathOperator{\cert}{C}
\DeclareMathOperator{\sens}{s}
\DeclareMathOperator{\UC}{UC}
\DeclareMathOperator{\D}{D}

\DeclareMathOperator{\md}{\delta}

\DeclareMathOperator{\C}{C}

\DeclareMathOperator{\haf}{\textnormal{\textsc{HAF}}}
\DeclareMathOperator{\maf}{\textnormal{\textsc{MAF}}}

\DeclareMathOperator{\chaf}{\textnormal{\textsc{CHAF}}}

\newcommand{\norm}[1]{\left\lVert#1\right\rVert}

\title{Boolean Functions with Minimal Spectral Sensitivity}
\author{Krišjānis Prūsis}
\author{Jevgēnijs Vihrovs}

\affil{Centre for Quantum Computer Science, Faculty of Science and Technology,\authorcr University of Latvia, Raiņa 19, Riga, Latvia, LV-1050}

\date{}

\begin{document}

\maketitle

\begin{abstract}
    We show examples of total Boolean functions that depend on $n$ variables and have spectral sensitivity $\Theta(\sqrt{\log n})$, which is asymptotically minimal.
    Our main new function combines the Hamming code with the Boolean address function and has $\lambda(f) = \sqrt{(1+o(1)) \log_2 n}$, which is optimal even up to a constant factor.
    By combining this function with itself in a specific way, we also obtain a family of functions with $\sens_0(f) = (c+o(1)) \log_2 n$ and $\sens_0(f) = (1-c+o(1)) \log_2 n$ for any $c \in [0,1]$.
    This is an optimal tradeoff for Boolean functions with low sensitivity, as the lower bound on sensitivity by Simon generalizes to
    \[\sens_0(f)+\sens_1(f)\geq\log_2 n - \log_2 \log_2 n + 2.\]
    As a corollary, this gives a new example of a function with minimal possible sensitivity (up to a constant factor), $\sens(f) = (\frac{1}{2}+o(1)) \log_2 n$.
\end{abstract}

\section{Introduction}

We investigate the question of the minimum possible spectral sensitivity of a total Boolean function $f$ that depends on all $n$ variables.
This question has been answered for other major measures, with the following lower bounds and matching examples:\footnote{Throughout the paper, we use $\log n$ to denote the logarithm base $2$.}
\begin{itemize}
    \item sensitivity: $\sens(f) \geq \frac{1}{2} \log n - \frac{1}{2} \log \log n + \frac{1}{2}$ \cite{Sim83, Weg85};
    \item degree: $\deg(f) \geq \log n - O(\log \log n)$ \cite{NS94};
    \item approximate degree: $\widetilde \deg(f) \geq \log n / \log \log n - O(1)$ \cite{AdW14}.
\end{itemize}
These in a sense were the lowest measures known, as they lower bound other main complexity measures for Boolean functions: the deterministic, randomized and quantum query complexities, block sensitivity, certificate complexity, etc., at the same time being polynomially related to them \cite{ABKRT21}.
Lower bounds for these measures were further studied in \cite{Wel22}.

Spectral sensitivity $\lambda(f)$ is a relatively new measure used to prove the Sensitivity Conjecture \cite{Hua19}, which asymptotically lower bounds all of $\sens(f)$, $\deg(f)$ and $\widetilde \deg(f)$ \cite{ABKRT21}.
In this work we show functions with $\lambda(f) = \Theta(\sqrt{\log n})$, which is asymptotically minimal.

In order to lower bound the minimal spectral sensitivity, we will use the following known relations. On the one hand, we have that the degree of a non-degenerate function is at least logarithmic:
\begin{lemma}[Theorem 1 in \cite{NS94}] \label{deg-bound}
For any non-degenerate Boolean function $f$,
$$\deg(f) \geq \log n - O( \log\log n).$$
\end{lemma}
We also have an upper bound on degree in terms of spectral sensitivity:
\begin{lemma}[Theorem 1 in \cite{ABKRT21}] \label{lower-bound}
    For any Boolean function $f$, $\deg(f)\leq \lambda(f)^2$.
\end{lemma}
Thus, any total non-degenerate function will have $\lambda(f) \geq \sqrt{(1+o(1))\log n}$.
On the other hand, there is the following upper bound on spectral sensitivity:
\begin{lemma}[Lemma 31 in \cite{ABKRT21}] \label{upper-bound}
For any Boolean function $f$, $\lambda(f)\leq \sqrt{\sens_0(f)\sens_1(f)}$.
\end{lemma}
Therefore, as sensitivity is lower bounded by $\Omega(\log n)$ \cite{Sim83}, one approach to construct a function with minimal spectral sensitivity is a function with $\sens_0(f)=1$ and $\sens_1(f)=O(\log n)$. 
For such a function, $\lambda(f) = \sqrt{\sens_1(f)}$ by the fact that $\sqrt{\sens_1(f)} = \sqrt{\sens(f)} \leq \lambda(f) \leq \sqrt{\sens_0(f)\sens_1(f)} = \sqrt{\sens_1(f)}$ by the above lemmas.

We give two methods to obtain such functions. The first is based on combining Hamming codes with the Boolean address function, and achieves $\lambda(f)=\sqrt{(1+o(1))\log n}$, an optimal constant factor. The second, using the desensitization of \cite{BHT17} on known low-sensitivity functions, is more general but achieves a worse constant: $\lambda(f) \geq \sqrt{(  3 + o(1)) \log n}$.

We also show that having one of $\sens_0(f)$, $\sens_1(f)$ be constant is not necessary to achieve this bound. We can modify the first construction to obtain a function with $\sens_0(f) = (c+o(1)) \log n$, $\sens_1(f)= (1-c+o(1)) \log n$ for any $0 \leq c \leq 1$ while maintaining $\lambda(f)=\sqrt{(1+o(1))\log n}$.
This is an optimal tradeoff, as we generalize the proof of the lower bound on sensitivity to imply a stronger statement:
\begin{reptheorem}{simon}[Implicit in \cite{Sim83}] 
For any non-degenerate Boolean function $f$,
\[\sens_0(f) + \sens_1(f) \geq \log n - \log \log n + 2.\]
\end{reptheorem}
In particular, we obtain a new example where the sensitivity of a Boolean function is minimal up to the constant, that is, $\sens_0(f) = \sens_1(f) = (\frac{1}{2} + o(1)) \log n$.
Up until now the only such example, to our knowledge, was the monotone address function \cite{Weg85}.

\section{Definitions}

Let $f: \{0,1\}^n \rightarrow \{0,1\}$ be a total non-degenerate Boolean function: for every $i\in[n]$, there is an input $x\in\{0,1\}^n$ such that $f(x) \neq f(x^{(i)})$, where $x^{(i)}$ is $x$ with the $i$-th variable flipped.

The \emph{sensitivity} of $f$ on an input $x$, denoted by $\sens(f,x)$, is defined as the number of $i \in [n]$ such that $f(x) \neq f(x^{(i)})$.
For $b \in \{0, 1\}$, the \emph{$b$-sensitivity} of $f$ is defined as $\sens_b(f) = \max_{x \in f^{-1}(b)} \sens(f,x)$.
The \emph{sensitivity} of $f$ is defined as $\sens(f) = \max\{\sens_0(f), \sens_1(f)\}$.

The \emph{sensitivity graph} of $f$ is defined as a subgraph $G_f=(V,E)$ of the $n$-dimensional Boolean hypercube, where $V = \{0,1\}^n$ and $E = \{\{x,x^{(i)}\} \mid f(x)\neq f(x^{(i)})\}$.
Let $A_f$ be the adjacency matrix of $G_f$; then the \emph{spectral sensitivity} of $f$, denoted by $\lambda(f)$, is defined as $\norm{A_f}$.
Since $G_f$ is bipartite and $A_f$ is real and symmetric, $\norm{A_f}$ is also the largest eigenvalue of $A_f$.

We say that an input $x$ \emph{satisfies} a partial assignment $p : [n] \to \{0,1,*\}$ and write $x \in p$ if for all $i\in[n]$ either $x_i = p(i)$ or $p(i) = *$.
The set of all inputs satisfying $p$ forms a subcube of the $n$-dimensional Boolean hypercube.
The dimension of this subcube is the number of $*$ in $p$, called the \emph{dimension} of $p$ and denoted by $\dim(p)$; the total number of $0$s and $1$s is called the \emph{co-dimension} of $p$ and denoted by $\codim(p) = n - \dim(p)$.

A \emph{certificate} is a partial assignment $p$ such that $f(x)$ is constant for all inputs $x \in p$.
The \emph{certificate complexity} of $x$ with respect to $f$ is defined as $\cert(f,x) = \min_{p : x \in p} \codim(p)$.
The \emph{$b$-certificate complexity} of $f$ is defined as $\cert_b(f) = \max_{x \in f^{-1}(b)} \cert(f,x)$.
The \emph{certificate complexity} of $f$ is defined as $\cert(f) = \max\{\cert_0(f), \cert_1(f)\}$.

A set of $b$-certificates $S$ such that each input $x \in f^{-1}(b)$ satisfies exactly one $p \in S$ is called an \emph{unambiguous collection} of $b$-certificates.
The \emph{$b$-unambiguous certificate complexity} of $f$ is defined as $\UC_b(f) = \min_S \max_{p \in S} \codim(p)$.
The \emph{unambiguous certificate complexity} of $f$ is defined as $\UC(f) = \max\{\UC_0(f), \UC_1(f)\}$.

\section{Hamming Address Function}

For the first approach, we directly construct a family of functions with minimal spectral sensitivity. The construction is a variation of the well known address function, where the address will be encoded using Hamming codes \cite{Ham50}. 

For any integer $r\geq 2$, the Hamming code $H_r$ encodes messages of length $2^{r}-r-1$ into codewords of length $2^{r}-1$ so that any two codewords differ in at least 3 positions. The total number of different messages for such a code is $|H_r|=2^{2^r-r-1}$.

\begin{theorem} \label{haddr}
    There is a family of total non-degenerate Boolean functions $\haf_r$ for $r\geq 2$ such that $\sens_0(\haf_r)=1$, $\sens_1(\haf_r) \leq 2^r$, $n \geq 2^{2^r-r-1}$. For these functions
    $\lambda(\haf_r)= \sqrt{(1 + o(1)) \log n}$.
\end{theorem}
\begin{proof}
    Let $k=2^{r}-1$. For any message $m \in [2^{2^r-r-1}]$ of $H_r$, denote by $w_m$ the corresponding codeword of length $k$.
    We construct the following $1$-certificate $p_m$ for $\haf_r$ for each such message:
    
    \[p_m=w_m *^{m-1} 1 *^{2^{2^r-r-1} - m}.\]

    We define $\haf_r(x)=1$ iff $x$ satisfies $p_m$ for some $m$.
    Each such $1$-certificate differs from any other in at least $3$ of the first $k$ positions, as they correspond to different codewords. Thus, any $0$-input can be adjacent to at most one such certificate and we have $\sens_0(\haf_r)=1$. 
    
    This function is non-degenerate, as each of the first $k$ positions is sensitive for every $1$-input, and each of the remaining $2^{2^r-r-1}$ positions is sensitive for $1$-inputs corresponding to one of these certificates.
    
    We have $\sens_1(\haf_r) \leq \C_1(\haf_r)=k+1=2^r$.
    Lastly, $n=k+2^{2^r-r-1} \geq 2^{2^r-r-1}$.
    Then for spectral sensitivity we have
    \[\lambda(\haf_r)=\sqrt{s_1(\haf_r)} = \sqrt{(1 + o(1)) \log n}.\qedhere\]
\end{proof}

\section{Desensitized Functions}
The second approach is to examine functions achieving $\sens_0(f)=\sens_1(f)=O(\log(n))$ and apply the desensitizing transformation of \cite{BHT17} to reduce their $0$-sensitivity to $1$. Although in general this transformation can asymptotically increase $1$-sensitivity, we show that this is not the case for some well-known functions with $\sens(f)=O(\log(n))$.
Some examples of such low-sensitivity functions include the Boolean address function, the balanced binary tree \cite{DS19} (Definition 3.1.), the function constructed in \cite{CHS20}, and the monotone Boolean address function \cite{Weg85}.

\begin{lemma}[Desensitization, Lemma 12 in \cite{BHT17}] \label{desens}
    For any $f : \{0,1\}^n \to \{0, 1\}$, there exists an $f' : \{0,1\}^{3n} \to \{0,1\}$ such that $\sens_0(f') = 1$ and $\UC_1(f') \leq 3\UC_1(f)$.
\end{lemma}

To bound the $1$-sensitivity of desensitized functions, we can observe that $\sens_1(f') \leq 3\D(f)$, since $\sens_1(f') \leq \UC_1(f') \leq 3\UC_1(f) \leq 3\D(f)$, which follows from the basic properties of $\UC$ \cite{BHT17}.
For each of the examples, we have $\D(f) = O(\log n)$, therefore we also have $\sens_1(f') = O(\log n)$.
Thus, $\lambda(f') = O(\sqrt{s_1(f')}) = O(\sqrt{\log n})$.

However, we will show that this approach gives us larger constant factors at $\sqrt{\log n}$ than the Hamming address function.
To do that, we can lower bound $\lambda(f')$ in terms of $\sens_1(f)$ or $\UC_1(f)$.
The desensitized function $f'$ in Lemma \ref{desens} is defined as follows.
Let $S$ be a collection of unambiguous $1$-certificates with co-dimension $\UC_1(f)$.
Define $x'=x_1x_2x_3$, then $f'(x') = 1$ iff there exists a $1$-certificate $p \in S$ such that $x_1, x_2, x_3 \in p$.
We can show that if $ppp(i) \neq *$, then $x'$ is sensitive on $i$.
If that was false, then $x'^{(i)}$ satisfies some $1$-certificate $qqq$, for $q \in S$ and $p \neq q$.
Since $p$ and $q$ differ in at least $1$ position, $ppp$ and $qqq$ differ in at least $3$ positions, but only a single position was changed, a contradiction.
Hence $\sens_1(f') \geq 3\UC_1(f)$ and $\sens_1(f') = \UC_1(f) = 3\UC_1(f)$.
By the fact that $\sens_0(f')=1$ and by Lemma \ref{lower-bound}, we have $\lambda(f') = \sqrt{\sens_1(f')} = \sqrt{3\UC_1(f)} \geq \sqrt{3\sens_1(f)}$.

All the above examples except the monotone address function have $\sens_1(f) = (1+o(1))\log n$.
In that case, $\lambda(f') \geq \sqrt{3\sens_1(f)} = \sqrt{(3+o(1))\log n}$.
The monotone address function $\maf_k$ \cite{Weg85}, however, has sensitivity $(\frac{1}{2}+o(1)) \log n$, so potentially $\lambda(\maf_k')$ could have a smaller constant factor than the other examples.
However, we show that is not the case.

This function is defined as follows.
There are $k$ address bits $x_1$, $\ldots$, $x_k$ and $m = \binom{k}{\lfloor k/2 \rfloor}$ data bits $y_1$, $\ldots$, $y_m$.
Take any bijection $g$ from $\{x \in \{0,1\}^n \mid |x| = \lfloor k/2 \rfloor\}$ to $[m]$.
Then
\[
\maf_k(x,y) =
\begin{cases}
    0, & \text{if $|x| < \lfloor k/2 \rfloor$,}\\
    y_{g(x)} & \text{if $|x| = \lfloor k/2 \rfloor$,}\\ 
    1, & \text{if $|x| > \lfloor k/2 \rfloor$.}
\end{cases}
\]
For this function, $\sens(\maf_k) = \lceil k/2 \rceil+1 = (1/2+o(1))\log n$.

\begin{proposition}
    $\lambda(\maf_k') = \sqrt{(3+o(1)) \log n}$.
\end{proposition}

\begin{proof}
    Examine the restriction of $\maf_k$ where all data bits $y_i$ are set to $0$.
    This restriction is just a threshold function on $k$ bits, and its degree is exactly $k$ (Proposition 5 in \cite{BdW02}).
    Since the degree is non-increasing under restrictions, we have $\deg(\maf_k) \geq k$.
    Since $\UC_1(f) \geq \UC_{\min}(f) \geq \deg(f)$ (Section 2.7.2.~in \cite{BHT17}), the required follows.
\end{proof}

\section{Sensitivity Tradeoff}

In this Section, we apply our $\haf_r$ function to obtain the optimal characterization of the relation between $0$- and $1$-sensitivities, in case when $\sens(f)$ is minimal.

We can notice that \cite{Sim83} actually gives a stronger result than stated, bounding the sum of $0$- and $1$- sensitivities. We restate the proof here for convenience.
\begin{theorem}[Implicit in \cite{Sim83}]  \label{simon}
For any non-degenerate Boolean function $f$,
\[\sens_0(f) + \sens_1(f) \geq \log n - \log \log n + 2.\]
\end{theorem}
\begin{proof}
    Let $i\in [n]$ be an arbitrary input position. Since $f$ is non-degenerate, there exists an input $x$ such that $f(x) \neq f(x^{(i)})$. W.l.o.g.~assume that $f(x)=0$ and $x_i=0$. For $b\in \{0,1\}$, let $C_b$ be the set of inputs $y$ with $y_i=b$. Now define subsets of inputs $U_b \subseteq C_b$; let $y\in U_0$ iff:
    \begin{itemize}
        \item there exists a path $p_0$ from $x$ to $y$ in $C_0$ such that every $y' \in p_0$ has $f(y')=0$;
        \item there exists a path $p_1$ from $x^{(i)}$ to $y^{(i)}$ in $C_1$ such that every $y'\in p_1$ has $f(y')=1$.
    \end{itemize}
    Let then $y\in U_1$ iff $y^{(i)} \in U_0$.
    Note that, for every $y\in U_0$, $f(y)=0,f(y^{(i)})=1$.
    Let $G_b$ be the subgraph of the hypercube induced by $U_b$ and denote by $\md(G_b)$ the minimum degree of this subgraph.

    Examine an input $y\in U_0$. It has sensitivity at most $s_0(f)$. Likewise, $y^{(i)}$ has sensitivity at most $\sens_1(f)$ such positions. As they are both sensitive on $i$, there are therefore at least $n-\sens_0(f)-\sens_1(f)+1$ positions $j$ such that $f(y^{(j)})=0$ and $f(y^{(i,j)})=1$. But then $y^{(j)} \in U_0$. Therefore $\md(G_0),\md(G_1) \geq n-\sens_0(f)-\sens_1(f)+1$.

    We next use the following lemma:
    \begin{lemma}[Lemma 1 in \cite{Sim83}]
    Let $G=(V,E)$ be a non-empty subgraph of the Boolean hypercube. Then $|V| \geq 2^{\md(G)}$. 
    \end{lemma}
    Therefore $|U_0|\geq 2^{ n-\sens_0(f)-\sens_1(f)+1}$. Note that each such input is sensitive to the $i$-th position, therefore there are at least that many edges in the $i$-th dimension in the sensitivity graph of $f$.
    We can repeat this argument for every $i\in[n]$, giving a total of at least $n2^{ n-\sens_0(f)-\sens_1(f)+1}$ such edges.

    On the other hand, the sensitivity graph of $f$ can have at most $\sens(f) 2^{n-1}$ edges. Therefore
    \begin{align*}
        \sens(f) 2^{n-1} &\geq n2^{ n-\sens_0(f)-\sens_1(f)+1} \\
        2^{\sens_0(f)+\sens_1(f)} & \geq \frac {4n} {\sens(f)}\\
        \sens_0(f)+\sens_1(f) &\geq \log n-\log \sens(f)+2.
    \end{align*}

    If $\sens(f) > \log n$, then $\sens_0(f)+\sens_1(f) > \sens(f) > \log n$ and we are done.
    Otherwise,
    \[\sens_0(f)+\sens_1(f) \geq \log n-\log \log n+2.\qedhere\]
    
\end{proof}

Therefore minimizing sensitivity up to constant factors requires a tradeoff between these two values.
We can see that both $\maf_k$ and $\haf_r$ achieve this tradeoff.
A natural question arises, whether there exists such a function in the general case.
Here we give a corresponding construction using $\haf_r$, answering the question in positive.
\begin{theorem}
    For any integer $k$ and any $0\leq c\leq 1$ there exists a total non-degenerate Boolean function $f$ on $n>k$ inputs such that $\sens_0(f)=(c+o(1))\log n$, $\sens_1(f)=(1-c+o(1))\log n$ and $\lambda(f)=\sqrt{(1+o(1))\log n}$.
\end{theorem}
\begin{proof}
    As the basis for our function we will take the composition of the Hamming address function with a negation of itself: $\haf_{r_1} \circ \neg \haf_{r_2}$.
    We perform two additional modifications to this.
    First, we do not compose the first $2^{r_1}-1$ bits of the outer function containing the codeword, but only the remaining $2^{2^{r_1}-r_1-1}$ data bits. Denote this new composition by $\circ'$.
    Then take $f=\haf_{r_1} \circ' \neg \haf_{r_2}$. 

    We have $\sens_0(\haf_{r_1})=1$. If this sensitive bit is in the codeword, the composition does not increase $0$-sensitivity, as we do not compose the codeword bits. If it is in a data bit, it will be multiplied by $\sens_0(\neg \haf_{r_2})=2^{r_2}$, giving $\sens_0(f)=2^{r_2}$.

    Conversely, $\sens_1(\haf_{r_1})=2^{r_1}$. However, only one of these sensitive positions will be a data bit and this data bit will have value $1$. But $\sens_1(\neg \haf_{r_2})=1$, therefore  $\sens_1(f)=2^{r_1}$.

    The total input size of $f$ is
    \[n=2^{r_1}-1+2^{2^{r_1}-r_1-1}\left(2^{r_2}-1+2^{2^{r_2}-r_2-1} \right). \]

    We can always add any value to both $r_1$ and $r_2$ to reduce the weight of lower order factors while maintaining the ratio between $s_0(f)$ and $s_1(f)$.
    Then $\log n = (1+o(1))(2^{r_1}+2^{r_2})=(1+o(1))(s_0(f)+s_1(f))$.
    This way we can also ensure $n>k$.
    There remains one problem: as both $s_0(f)$ and $s_1(f)$ are powers of $2$, we can construct such $f$ only for $c$ such that $c/(1-c) = 2^a$ for integer values of $a$.

\bigskip

    To remedy this, the second change we make is replacing the Hamming address function with a conjunction of several. Take a set of Hamming codes $H_{r_1}, \ldots, H_{r_l}$. A set of valid codewords for each of these codes together encodes a message from a set of size $t=2^{2^{r_1}-r_1-1} \cdot \ldots \cdot 2^{2^{r_l}-r_l-1}= (1+o(1))2^{2^{r_1}+\ldots+2^{r_l}} $. Then, as with the Hamming address function, for each message $m\in [t]$ we construct a 1-certificate $p_m$ with the corresponding codewords in the first $2^{r_1}+\ldots+2^{r_l}-l$ positions and a $1$ in the $m$-th of the remaining $t$ positions.
    
    Then define the conjunction Hamming address function: $\chaf_{r_1, \ldots, r_l}(x)=1$ iff $x$ satisfies one of these $p_m$.
    As each pair of these certificates will have different codewords in at least one of the Hamming codes, it is easy to check that $\sens_0(\chaf_{r_1, \ldots, r_l})=1$ and $\sens_1(\chaf_{r_1, \ldots, r_l})=2^{r_1}+\ldots+2^{r_l}-l+1$.
    We can now replace the functions in our composition with this conjunction version to obtain $f'=\chaf_{a_1,\ldots,a_l} \circ' \neg\chaf_{b_1, \ldots, b_m}$.

    Similarly to $f$,  we obtain $\sens_0(f')=2^{b_1}+\ldots+2^{b_m}-m+1$ and $\sens_1(f')=2^{a_1}+\ldots+2^{a_l}-l+1$.
    By adding the same value $g$ to all of $a_1,\ldots,a_l,b_1,\ldots, b_m$ we can again discard lower order factors to obtain a function with $\log n = (1+o(1))(s_0(f')+s_1(f'))$. By choosing appropriate values for $a_1,\ldots,a_l,b_1,\ldots, b_m$ we can obtain $s_0(f')/s_1(f') \to c/(1-c)$ for any $c$, as $g \to \infty$.

\bigskip     It remains to determine the value of $\lambda(f')$. It suffices to examine the connected components of the sensitivity graph separately. It turns out there are only two types of such components. Examine a $1$-input $x$ of $f'$. It will be sensitive to every address bit in the outer $\chaf$. The $0$-inputs obtained by flipping such an address bit will have a sensitivity of $1$. Additionally, a $1$-input for $f'$ must have a $0$-input for the inner $\chaf$ at the position pointed to by the address of the outer $\chaf$. 
    
    If this $0$-input of the inner function has no sensitive bits, the examined component is just a star graph of degree $s_1(f')-1$ with $x$ at the center. 
    For this star graph, the spectral norm of its adjacency matrix is $\sqrt{s_1(f')-1}$.

    If this $0$-input of the inner function has a sensitive bit, denote the input obtained by flipping this bit by $y$. Note that $f'(y)=0$ and $s(f',y)=s_0(f')$, as it is sensitive to changing the value of the inner function back to $0$. Each of the resulting $1$-inputs, including $x$, will have sensitivity $s_1(f')$. But they all share a single adjacent $0$-input with sensitivity greater than $1$, $y$. Therefore we obtain a two-layer star graph with $y$ at the center with degree $s_0(f')$ and the first layer consists of vertices with degree $s_1(f')$, see Figure \ref{fig:sens}.
    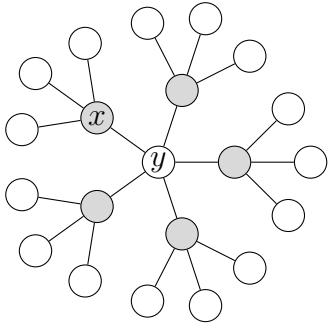
\begin{figure}[ht]
    \centering
    \begin{tikzpicture}[node distance=1.5cm, every node/.style={circle, draw, minimum size=4.25mm,inner sep=0}]

    \node (central) at (0,0) {$y$};
    
    \foreach \i in {1,...,5} {
        \ifnum\i=2
            \node[fill=gray!30] (layer1-\i) at (\i*72:1cm){$x$};
        \else 
            \node[fill=gray!30] (layer1-\i) at (\i*72:1cm){};
        \fi
        
        \draw (central) -- (layer1-\i); 
    }
    
    \foreach \i in {1,...,5} {
        \foreach \j in {1,...,3} {
            \node (layer2-\i-\j) at ($ (layer1-\i) + (\i*72+\j*45-90:1cm) $){};
            \draw (layer1-\i) -- (layer2-\i-\j); 
        }
    }
    
    \end{tikzpicture}
    \caption {The component type for the sensitivity graph of $\chaf_{a_1,\ldots,a_l} \circ' \neg\chaf_{b_1, \ldots, b_m}$ achieving maximal spectral sensitivity.}
    \label{fig:sens}
    \end{figure}
    
    Note that, for any eigenvalue of this component's adjacency matrix, there will be an eigenvector with equal weights for every vertex in each layer -- we can obtain such an eigenvector from a non-symmetric one by averaging it over all permutations of the children of each vertex (taking $y$ as the root). Then such an eigenvector is characterized by three parameters $w_y, w_1, w_2$, the weight at $y$ and the weight at the first and second layers of the star. If $w_2=0$, then so are $w_1,w_y$ and we get a $0$ vector. Otherwise, we can normalize it so that $w_2=1$. Then we obtain the following system of equations:
    \[
    \begin{cases}
        s_0(f')\cdot w_1=\lambda w_y, \\
        s_1(f')-1+w_y=\lambda w_1, \\
        w_1=\lambda.
    \end{cases}
    \]
Solving this system gives $\lambda(f')=\sqrt{s_0(f')+s_1(f')-1}=\sqrt{(1+o(1))\log n}$.
There are two solutions with an eigenvalue of this magnitude, one negative and one positive.
\end{proof}
\section{Acknowledgements}

We thank Manaswi Paraashar for introducing us to this problem, and Juris Smotrovs and Andris Ambainis for helpful discussions.

This work was supported by the Latvian Quantum Initiative under European Union Recovery and Resilience Facility project no.~2.3.1.1.i.0/1/22/I/CFLA/001.

\printbibliography

\end{document}